\documentclass[11pt]{article}

\usepackage{pdfpages}
\usepackage{fullpage}
\usepackage{times}
\usepackage{graphicx,color}
\usepackage[export]{adjustbox}
\usepackage{array,float}
\usepackage{xcolor}
\definecolor{devpubscyan}{HTML}{0066CC}
\usepackage{hyperref}

\newcolumntype{L}[1]{>{\raggedright\let\newline\\\arraybackslash\hspace{0pt}}m{#1}}
\newcolumntype{C}[1]{>{\centering\let\newline\\\arraybackslash\hspace{0pt}}m{#1}}
\newcolumntype{R}[1]{>{\raggedleft\let\newline\\\arraybackslash\hspace{0pt}}m{#1}}

\DeclareUrlCommand{\url}{%
    
    }

\usepackage[numbers]{natbib}
\usepackage{amstext,amssymb,amsmath}
\usepackage{amsthm}
\usepackage{algorithm}
\usepackage[noend]{algorithmic}
\algsetup{linenodelimiter=.}
\usepackage{mathtools}
\usepackage{verbatim}
\usepackage{bm}
\usepackage{bbm}
\usepackage{paralist}
\graphicspath{{graphs/}}
\usepackage{caption}
\usepackage{subcaption}
\usepackage{listings}
\usepackage{textcomp}
\usepackage{enumitem}

\usepackage{cleveref}

\newtheoremstyle{dotless}{}{}{\itshape}{}{}{. }{ }{}
\theoremstyle{dotless}
\newtheorem{lem}{Lemma}[section]

\newtheorem{thm}{Theorem}[section]

\newtheorem{defn}[lem]{Definition}

\newcommand{\A}{\mathcal{A}}
\newcommand{\eps}{\epsilon}

\newcommand{\E}{\mathbb{E}}

\newcommand{\cD}{\mathcal{D}}
\newcommand{\cY}{\mathcal{Y}}
\newcommand{\cZ}{\mathcal{Z}}

\newcommand{\bbv}{\pmb{v}}

\newcommand{\AclientOHE}{\A_{\sf client-OHE}}
\newcommand{\AserverOHE}{\A_{\sf server-OHE}}
\newcommand{\AclientsymOHE}{\A_{\sf client-symOHE}}
\newcommand{\AserversymOHE}{\A_{\sf server-symOHE}}

\newcommand{\aggregator}{\texttt{Aggregator}}

\newcommand{\var}[1]{\mathbf{Var}\left[#1\right]}

\newcommand{\reference}{\mathcal{R}}

\author{\hspace{1in}Audra McMillan\thanks{Apple} \footnote{Corresponding author: \href{mailto:audra.mcmillan@apple.com}{audra.mcmillan@apple.com}}
\and Omid Javidbakht\footnotemark[1]
\and Kunal Talwar\footnotemark[1] \hspace{1in}
\and Elliot Briggs\footnotemark[1]
\and Mike Chatzidakis\footnotemark[1]
\and Junye Chen\footnotemark[1]
\and John Duchi\thanks{Stanford University and Apple}
\and Vitaly Feldman\footnotemark[1]
\and Yusuf Goren\footnotemark[1]
\and Michael Hesse\footnotemark[1]
\and Vojta Jina\footnotemark[1]
\and Anil Katti\footnotemark[1]
\and Albert Liu\footnotemark[1]
\and Cheney Lyford\footnotemark[1]
\and Joey Meyer\footnotemark[1]
\and Alex Palmer\footnotemark[1]
\and David Park\footnotemark[1]
\and Wonhee Park\footnotemark[1]
\and Gianni Parsa\footnotemark[1]
\and Paul Pelzl\footnotemark[1]
\and Rehan Rishi\footnotemark[1]
\and Congzheng Song\footnotemark[1]
\and Shan Wang\footnotemark[1]
\and Shundong Zhou\footnotemark[1] }

\title{Private Federated Statistics in an Interactive Setting}
\date{}

\begin{document}
\maketitle

\begin{abstract}
Privately learning statistics of events on devices can enable improved user experience. Differentially private algorithms for such problems can benefit significantly from interactivity. We argue that an aggregation protocol can enable an interactive private federated statistics system where user's devices maintain control of the privacy assurance. We describe the architecture of such a system, and analyze its security properties.
\end{abstract}

\section{Introduction}

Gaining insight into trends in the user population is an important aspect of improving the user experience. For example,
learning popular words that users type, helps improve keyboard language models. The data needed to derive these insights can be personal and sensitive, and must be kept private. We study the problem of learning aggregate information about data that lives on device, while providing strong privacy protection to individuals. Often these data are multi-dimensional, and one wants to learn, for example, correlations between various attributes. Traditional approaches to this problem such as RAPPOR~\cite{erlingsson2014rappor} and others (e.g. ~\cite{applelearningatscale}) convert the high-dimensional data to categorical data (e.g. by bucketing), and use algorithms such as binary randomized response  \cite{W65} (2RR) on a one-hot encoding of the data, or refinements of it to build a noisy histogram. These systems are often static; they are designed to, statically, send snapshots of the data at a pre-specified cadence.

In this work, we propose an iterative approach to such problems, that aligns better with how data analysts typically interact with data. Such an interactive approach can also allow for more accurate answers to queries of interest. We describe how such an approach can be implemented and consider the privacy risks against different kinds of attackers.
An aggregation protocol ensures that the server only sees the \emph{aggregate} of the reports sent by users. We use a locally differentially privacy (DP) algorithm to add noise to the reports, so that the privacy of the final aggregate can be analyzed using recent upper bounds for privacy amplification by shuffling \cite{CheuSUZZ19, ErlingssonFMRTT19}.

One of the main challenges with implementing a federated statistics platform that allows for interaction is ensuring that the device can enforce any privacy constraints; both individual privacy constraints and ensuring that the user's data is only used for appropriate queries. To address individual privacy constraints, we propose an on-device differential privacy budget accountant that allows the device to limit the total individual privacy risk of privatized data leaving the device. We also propose an on-device query verification framework that validates that the user's data is only used in queries that are part of an approved query class. More details on both of these key features are given in Section~\ref{systemsoverview}. We will also discuss additional concerns that may arise when implementing such a system, including secure data storage and auditability.

At a high level, interactive algorithms for histogram estimation build on top of a private histogram algorithm such as RAPPOR~\cite{erlingsson2014rappor} or PI- RAPPOR ~\cite{pmlr-v139-feldman21a}. However, instead of a fixed histogram of the data, the algorithm can iteratively specify how the data should be bucketed to define the histogram cells. For example we may first build a coarse grained histogram of the data. Based on the measurements thus made, the algorithm can refine the bucketing and build a new histogram.
Each such query gets answered by devices using a local DP algorithm, coupled with an aggregation framework such as PRIO~\cite{Prio}. Privacy amplification by shuffling bounds allow us to control the privacy loss of one such query.
Results bounding the privacy loss of composing multiple differentially private algorithms can allow us to ensure that the total privacy cost, of the sequence of adaptive queries, against the server can be controlled.
We propose a system design that makes it possible for the user’s device to verify various parameters of the queries and ensure a pre-specified privacy budget. Our system starts with a specification of the set of allowed queries and privacy budgets for a particular analysis. Our design ensures that the device can enforce these constraints. Indeed one of our design goals is that an inadvertent error by the analyst translates to loss of utility for the analyst, and cannot lead to additional privacy loss.

This interactive approach offers more accurate answers than a more traditional static approach in many settings. The benefit of interactive queries in such contexts is well-studied; indeed in the central model of differential privacy, there is a long line of work on adaptive approaches to building histograms~(e.g. \cite{Cormode2012DifferentiallyPS, Qardaji:2012, Song:2013, Bagdasaryan:2022}). As a simple example, suppose that we wish to learn new n-grams of words that are typed on devices. In a traditional static system, we would first decide on the length n of the n-grams we want to learn in advance, say n = 4. Each device would pick a random 4-gram that they typed, and send a local DP report based on it. We would thus build a noisy histogram of the frequency counts of the 4-grams chosen by the devices. In an interactive approach, we would first construct a histogram of 1-grams, i.e. of the popular words. We then use the knowledge of the frequent 1-grams to restrict our search for 2-grams; we only attempt to learn 2-grams that start with a frequent 1-gram. The frequent 2-grams are then extended to 3-grams and so on. An obvious advantage of this approach, and one exploited in previous work in the central model, is that we learn a hierarchical representation of the data and e.g. get the frequent 3-grams which may not be prefixes of any
frequent 4-gram. Note that estimating the frequency of a 3-gram by adding up frequencies of all 4-grams that start with it is not a feasible option in the DP setting as each of these frequencies has some noise added, leading to a much larger standard deviation for the noise for the sum than is required. The limit on the number of frequent 3-grams also limits the number of possible 4-grams we potentially report on. Thus the “dictionary’‘ over which we run a histogram query is smaller, allowing for higher accuracy and efficiency. In addition, in the event that user’s have multiple 4-grams on their device, this approach makes it feasible to select which 4-gram to report on, based on the learnt information. More generally, the device may report on a subset of 4-grams~\cite{Kim:2021, pmlr-v119-gopi20a} and restricting oneself to 4-grams whose prefix 3-grams are frequent in the dataset can lead to more accurate answers.
Such an approach works broadly for any hierarchically structured data universe, of which n-grams over words is but one example. Geographical data falls most naturally in this class and has been used for visualizations in several works, see, for example~\cite{Erlingsson2020}. Understanding the distribution of a single real-valued feature also falls in this category: we would like to bucket the data into intervals that are as fine as can be supported by noisy histogram measurements. When the distribution is unknown, we can learn the right bucketing by following such an adaptive algorithm, repeatedly refining buckets with large counts.
We will consider learning new $n$-grams (sequences of $n$ words) typed by users using the keyboard as a running example throughout this paper.

In addition to greater algorithmic flexibility, an interactive approach is more compatible with how a data analyst would typically study a dataset. They may start with looking at the one-dimensional marginals of the data at some granularity, and then want to look at 2-way marginals, or for a bucketization of the data that depends on the distribution.

\section{Related Work}

Differential privacy \citep{DMNS06} began in the cryptography and theoretical computer science communities over a decade ago and has become the privacy benchmark across a variety of fields \cite{DR14}.  In this work, we will consider the local model of differential privacy \cite{Kasiviswanathan:2008}, as well as the aggregate model of differential privacy, which is closely related to the recently introduced shuffle model \cite{ErlingssonFMRTT19, CheuSUZZ19}.

There has been considerable work on differential privacy in the local model, first introduced in \cite{Kasiviswanathan:2008}.
This work includes obtaining tight bounds for statistical estimators \citep{duchi2013local}, learning problems \citep{Kasiviswanathan:2008}, and studying frequency estimation \citep{BS15,erlingsson2014rappor,FPE16}. The main frequency estimation algorithm of interest in this paper, private one-hot encoding, was proposed and used in RAPPOR \cite{erlingsson2014rappor}. Several follow-up works have reduced the communication, time and space complexity of private frequency estimation, or the problem of privately identifying heavy hitters, while maintaining accuracy~\cite{BS15, MS06, hsu2012distributed, pmlr-v139-feldman21a}. While we will not discuss these techniques in this paper, they are certainly of interest in a system like the one we describe.

The idea of using an aggregation protocol as a subroutine goes back to the early days of differential privacy~\cite{ODO}. It has been proposed for smart meter data privacy in~\cite{Aces:2011}. As it forms the central primitive in our proposed system, we name this model Aggregate Differential Privacy. A closely related model is the shuffle model, which was first introduced in two concurrent works
\citep{ErlingssonFMRTT19,CheuSUZZ19}. These works demonstrated that shuffling data sent from users using locally differentially private randomisers can provably amplify DP guarantees. Since then a number of works have studied privacy amplification by shuffling and the model augmented with a shuffler more generally (e.g.~\citep{Balle:2019,Ghazi:2019,GhaziPV19,Balle:2020,Balle2020,GMPV20,cheu2020limits,CheuSUZZ19,BalcerCheu,WangXDZHHLJ20,Erlingsson2020,FeldmanMT:2020,girgis2020shuffled}). This amplification by shuffling can also be used to analyze models augmented with an aggregation protocol such as PRIO \citep{Prio}
(since the sum of real values provides even less information than shuffled values). In particular, privacy amplification by shuffling was used in Apple’s and Google’s Exposure Notification Privacy-preserving Analytics~\citep{ENPA:2021}.

The power of interaction in locally differentially private algorithms has been well documented in the literature \cite{Joseph:2019, 10.5555/3454287.3455630, Dagan:2020} by works showing that the ability to adaptively query users results in strictly more accurate algorithms. Recent work has also explored the power of interactivity in the setting of aggregate or shuffle differential privacy. Bagdasaryan et. al \citep{Bagdasaryan:2022} propose using an interactive algorithm similar to the one described in the introduction to learn location heatmaps under aggregate differential privacy. They demonstrate the improvement of this algorithm over standard non-interactive approaches.

There have been a few experimentaly systems and practical deployments of differentially private systems. Unlike our system, most of these employ central differential privacy, including the PSI ($\Psi$) system \cite{Psi}, OpenDP \cite{OpenDP}, Census Bureau system \cite{Census2, census2020}, Airavat \cite{RSKSW}, PINQ \cite{PINQ}, and GUPT \cite{MTSSC12}. In the local setting, RAPPOR deploys differential privacy on the Chrome web browser \cite{erlingsson2014rappor,FPE16}.

\section{Mathematical Preliminaries}
Consider a data universe $\cD$ which has $p$ elements, typically written as $[p] = \{ 1, \cdots, p\}$. Let $n$ be the total number of records, and for record $i \in [n]$, denote its data entry as $d^{(i)} \in \cD$. 

\subsection{Differential Privacy}

In this work, we describe a system with two levels of privacy protection; differential privacy in the aggregate model and differential privacy in the local model. The system is designed to provide the first level of privacy through the use of an aggregation protocol. An aggregation protocol ensures that the server only has access to the sum of the individual reports sent from each device. The aggregate differential privacy guarantees are achieved after the data leaves the aggregation protocol. As a secondary layer of privacy protection, local differential privacy guarantees are enforced on device: any sensitive data leaving a user's device will be protected by local differential privacy. We will introduce these two types of privacy guarantees in this section.

\subsubsection{Local Differential Privacy}

Let us first introduce local differential privacy. This privacy guarantee is enforced on-device via the use of local randomizers.

\begin{defn}[Local Randomizer \cite{DR14,Kasiviswanathan:2008,Erlingsson2020}]\label{LDP}
Let $\A: \cD \to \cY$ be a randomized algorithm mapping a data entry in $\cD$ to an output space $\cY$.  
The algorithm $\A$ is an $\epsilon$-DP local randomizer if for all pairs of data entries $d,d'\in\cD$, and all events $E\subset\cY$, we have $$
- \epsilon \leq \ln\left(\frac{\Pr[\A(d) \in E ]}{\Pr[\A(d') \in E ]}  \right)\leq \epsilon.
$$
\end{defn}
\noindent The privacy parameter $\eps$ captures the \emph{privacy loss} consumed by the output of the algorithm. If $\eps = 0$ then the output is independent of the input and we are ensuring perfect privacy, while $\eps = \infty$ enforces no constraints on $\A$ and hence provides no privacy guarantee. Differential privacy for an appropriate $\eps$ ensures that it is impossible to confidently determine what the individual contribution was, given the output of the mechanism. In turn, this implies a strong Bayesian interpretation: the posterior given the output of the mechanism is point-wise multiplicatively close to the prior one may have about an individual's data.

In general, differential privacy is defined for algorithms with input databases with more than one record. In the local model of differential privacy, algorithms may only access the data through a local randomizer so that no raw data leaves the device. For a single round protocol, local differential privacy is defined as follows:
\begin{defn}[Local Differential Privacy \cite{Kasiviswanathan:2008}]\label{localDP}
Let $\A: \cD^n \to \cZ$ be a randomized algorithm mapping a dataset with $n$ records to some arbitrary range $\cZ$.  The algorithm $\A$ is $\epsilon$-local differentially private if it can be written as $\A(d^{(1)}, \cdots, d^{(n)}) = \phi\left(\A_1(d^{(1)}), \cdots, \A_n(d^{(n)}) \right)$ where the $\A_i: \cD \to \cY$ are $\epsilon$-local randomizers for each $i \in [n]$ and $\phi: \cY^n \to \cZ$ is some post-processing function of the privatized records $\A_1(d^{(1)}),\cdots, \A_n(d^{(n)})$.  Note that the post-processing function does not have access to the raw data records.
\end{defn}

\noindent Local differential privacy guarantees are enforced on-device. 

The definition of a local randomizer given in Definition~\ref{LDP} is sometimes called the \emph{replacement model} of local differential privacy. A slightly weaker notion referred to as the deletion model \cite{Erlingsson2020} can also be used to provide meaningful privacy guarantees with slightly more accurate results when local differential privacy is the primary privacy guarantee. Since we are concerned with aggregate differential privacy in this paper, we will focus on the replacement model, which is more appropriate for combining with aggregate differential privacy.

\subsubsection{Aggregate Differential Privacy}

The aggregate model of differential privacy is a distributed model of computation in which, as in the local model, clients hold their own data and a server communicates with the clients in a federated manner to perform data analysis. In addition, the model includes a \emph{aggregation protocol} that sums the local reports, with the guarantee that the output of the aggregation protocol does not reveal anything about the local reports except their sum\footnote{The term aggregation protocol is commonly used to refer to protocols that perform other types of aggregation. In this paper, we will only be referring to aggregation protocols that output the sum of the input values}. The aggregation protocol computes the sum of the local reports, and this is released to the server\footnote{The aggregate model of DP is a special case of the more general and common shuffle model of differential privacy introduced in \citep{ErlingssonFMRTT19, CheuSUZZ19} which permutes the local reports, rather than summing. Since the system we propose uses an aggregation protocol, we focus on the aggregate model of DP in this paper. The sum of the local reports reveals strictly less information than a permutation of the local reports, so the aggregate model enjoys stronger privacy guarantees.}

\begin{defn}\label{defnaggregateDP}
A single round algorithm $\A$ is $(\eps, \delta)$-DP in the aggregate model if the output of the aggregation protocol on two datasets that differ on the data of a single individual are close. Formally, an algorithm $\A:\cD^n\to\cZ$ is $(\eps, \delta)$-DP in the aggregate model if the following conditions both hold:
\begin{itemize}
\item it can be written as $\A(d^{(1)}, \cdots, d^{(n)})=\phi(\aggregator(f(d^{(1)}), \cdots, f(d^{n)}))$ where $f:\cD\to\cZ$ is a randomized function that transforms that data, $\aggregator$ is an aggregation protocol, and $\phi:\cY^n\to\cZ$ is some post-processing of the aggregated report
\item for any pair of datasets $D$ and $D'$ that differ on the data of a single individual, and any event $E$ in the output space,
\[
\Pr(\A(D)\in E)\le e^{\eps} \Pr(\A(D')\in E)+\delta.
\]
\end{itemize}
Note that the post-processing function takes the aggregation as its input and does not have access to the individual reports.
\end{defn}

A multi-round algorithm $\A$ is $(\eps,\delta)$-DP in the aggregate model if it is the composition of single round algorithms which are DP in the aggregate model, and the total privacy loss of $\A$ is $(\eps,\delta)$-DP.
One can formulate a version of Definition~\ref{defnaggregateDP} for multi-round algorithms, for a more in-depth discussion see \cite{Jain2021ThePO}. In this paper, we will assume that the privacy guarantee of multi-round algorithms is computed using standard theorems that bound the privacy guarantee of the adaptive composition of multiple differentially private single round algorithms~\cite{DMNS06, DRV}. This is the case for most algorithms of interest in this setting.
In the system described in this paper, we will always use functions $f$ that are $\eps_0$-DP local randomizers. This ensures that all our algorithms are $\eps_0$-differentially private in the local model \emph{and} $(\eps,\delta)$-differentially private in the aggregate model. When each user uses a DP local randomizer to communicate their data, the privacy guarantee in the aggregation model, $\eps$, can be bounded by a quantity that is a function of both $\eps_0$, the privacy guarantee in the local model, and $n$, the number of users that participate in the aggregation protocol. The following result, commonly referred to as \emph{privacy amplification by aggregation}, shows how $\eps$ decreases as $n$ increases. This phenomena was originally observed by \citep{ErlingssonFMRTT19} and \citep{CheuSUZZ19}, although the version we state here is due to \cite{Feldman:2022}.

\begin{thm}\label{amplificationbyaggregation}
Suppose $\A:\cD^n\to\cZ$ is an algorithm that can be written as $\A(d^{(1)}, \cdots, d^{(n)})=\phi(\aggregator(f(d^{(1)}, \cdots, f(d^{n)}))$ where $f:\cD\to\cZ$ is a randomized function that transforms that data, $\aggregator$ is an aggregation protocol, and $\phi:\cY^n\to\cZ$ is some post-processing of the aggregated report. If $f$ is an $\eps_0$-DP local randomizer in the replacement model, then for any $\delta\in[0,1]$ such that $\eps_0\le\ln(\frac{n}{8\ln(2/\delta)}-1)$, $\A$ is $(\eps,\delta)$-DP in the aggregate model where
\begin{equation}\label{shufflebound}
\eps \le \log\left(1+(e^{\eps_0}-1)\left(\frac{4\sqrt{2\log(4/\delta)}}{\sqrt{(e^{\eps_0}+1)n}}+\frac{4}{n}\right)\right).
\end{equation}
\end{thm}

Note that Theorem~\ref{amplificationbyaggregation} immediately implies a corresponding theorem for the case where the functions $f$ are $\eps_0$-DP local randomizers in the deletion model by replacing $\eps_0$ with $2\eps_0$ in eqn~\eqref{shufflebound}. 
This upper bound on $\eps$ is approximately $O\left(\eps_0\frac{\sqrt{\log(1/\delta)}}{\sqrt{n}}\right)$ when $\eps<1$ and $O\left(\frac{\sqrt{e^{\eps_0}\log(1/\delta)}}{\sqrt{n}}\right)$ when $\eps\ge1$, where $\eps_0$ is the local privacy parameter in the replacement model. Numerical computations of the privacy amplification can obtain tighter bounds on the privacy of the output of the aggregation protocol. In the system described in this paper, we propose computing the privacy amplification numerically using the technique presented in \cite[Theorem 3.1]{Feldman:2022}. This code is available at \url{https://github.com/apple/ml-shuffling-amplification}. In Figure~\ref{varyingnamplification} we show how the privacy of the output of the aggregation protocol changes as a function of the cohort size for two given local privacy guarantees $\eps_0=3$ and $\eps_0=6$. For example, we can see that if each user sends their data using a $3$-DP local randomizer, then the privacy of the output of the aggregation protocol is at most $(1,10^{-6})$-DP in the aggregate model once $n>1,000$ and decreases quickly as the cohort size grows.

\begin{figure}
\begin{center}
\includegraphics[scale=0.5]{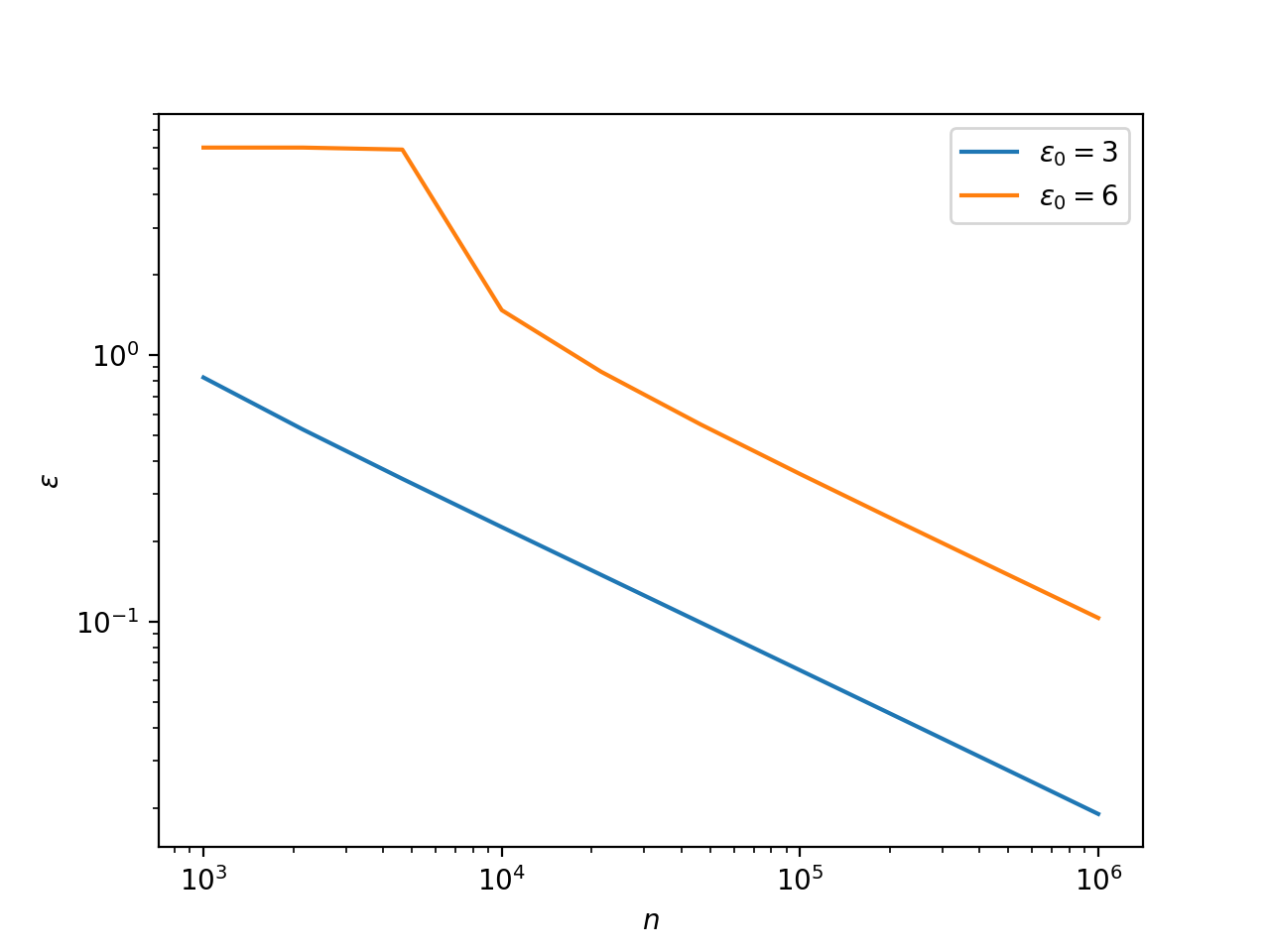}
\end{center}
\caption{Cohort aggregate privacy guarantees ($(\eps, 10^{-6})$) as a function of cohort size ($n$) for $\eps_0=3$ and $\eps_0=6$.}
\label{varyingnamplification}
\end{figure}

We remark that one can also measure the privacy amplification in terms of an alternative definition of privacy known as R\'enyi differential privacy~\cite{Mironov2017RnyiDP}. Using R\'enyi differential privacy will typically allow us to get stronger bounds for the privacy guarantee of composing multiple differentially private algorithms. We refer the reader to~\cite{Feldman:2022, Kim:2021} for details, and for simplicity, assume here that we compose $(\eps,\delta)$-DP guarantees using the advanced composition theorem~\cite{DRV}.

Since we want the output of the aggregation protocol to achieve the required level of privacy regardless of the cohort size, we suggest using an aggregation protocol that has the additional functionality that given a minimum cohort size $m$, the protocol will only release the aggregate if the actual cohort size is greater than $m$. If too few users participate in answering a query, the aggregation protocol does not output anything.  Given a local privacy parameter $\eps_0$, and a desired privacy guarantee on the output of the aggregation protocol $\eps$, we can compute the minimum cohort size $m$ for which the privacy guarantee would hold if $m$ users were to participate. The minimum cohort is communicated from the device to the aggregation protocol.
We will refer to the privacy guarantee of the output of the aggregation protocol with this additional functionality as the \emph{cohort-aggregate privacy guarantee}.

Our privacy analysis of the aggregate model via Theorem~\ref{amplificationbyaggregation} implicitly relies on each user gaining privacy from the noise added to the aggregate by all the other users. While we assumed that all clients are honest, this can be easily relaxed to assuming that at least a constant fraction of clients are honest and applying the privacy amplification by aggregation to the subset of honest clients.
An alternative approach to reducing the trust in user's behaving honestly would be to equip the aggregation protocol with the ability to add the noise needed for privacy.

\section{System Overview}\label{systemsoverview}

Our proposed architecture consists of device- and server- side infrastructure. We will use the term \emph{query} to refer to a single data collection event on device,
and \emph{analysis} to refer to a collection of queries performed for a given task.
At a high level,
the device stores information about the kind of queries that are allowed in the analysis as well as their privacy budgets. This information is hard-coded as part of the device OS, and encodes the immutable parts of the policy that ships with the OS. The data, as it is generated, get saved in a secure data storage. The server infrastructure allows an analyst to send interactive queries to devices, which are then verified (that they are allowed queries) on device. The response, with noise added, gets sent to an aggregation protocol, which allows the analyst to see the aggregate results of the query. We provide details of various parts of the system in this section.

Our proposed design allows for interactive queries. However, the queries are constrained to belong a certain class of queries and a description of this query class gets stored on device. This description includes the set of attributes that the analysis is allowed to access, as well as the kinds of queries that are permitted. We defer more details about the query class constraints and verification to Section~\ref{queryverification}.

 In addition, the answers to these queries are governed by a set of privacy constraints.
We impose global constraints for the whole analysis as well as per-data field privacy constraints. Section~\ref{sec:privacy-budget} has further details on the privacy constraints.

\begin{figure}
	\begin{center}
			\includegraphics[scale=0.15]{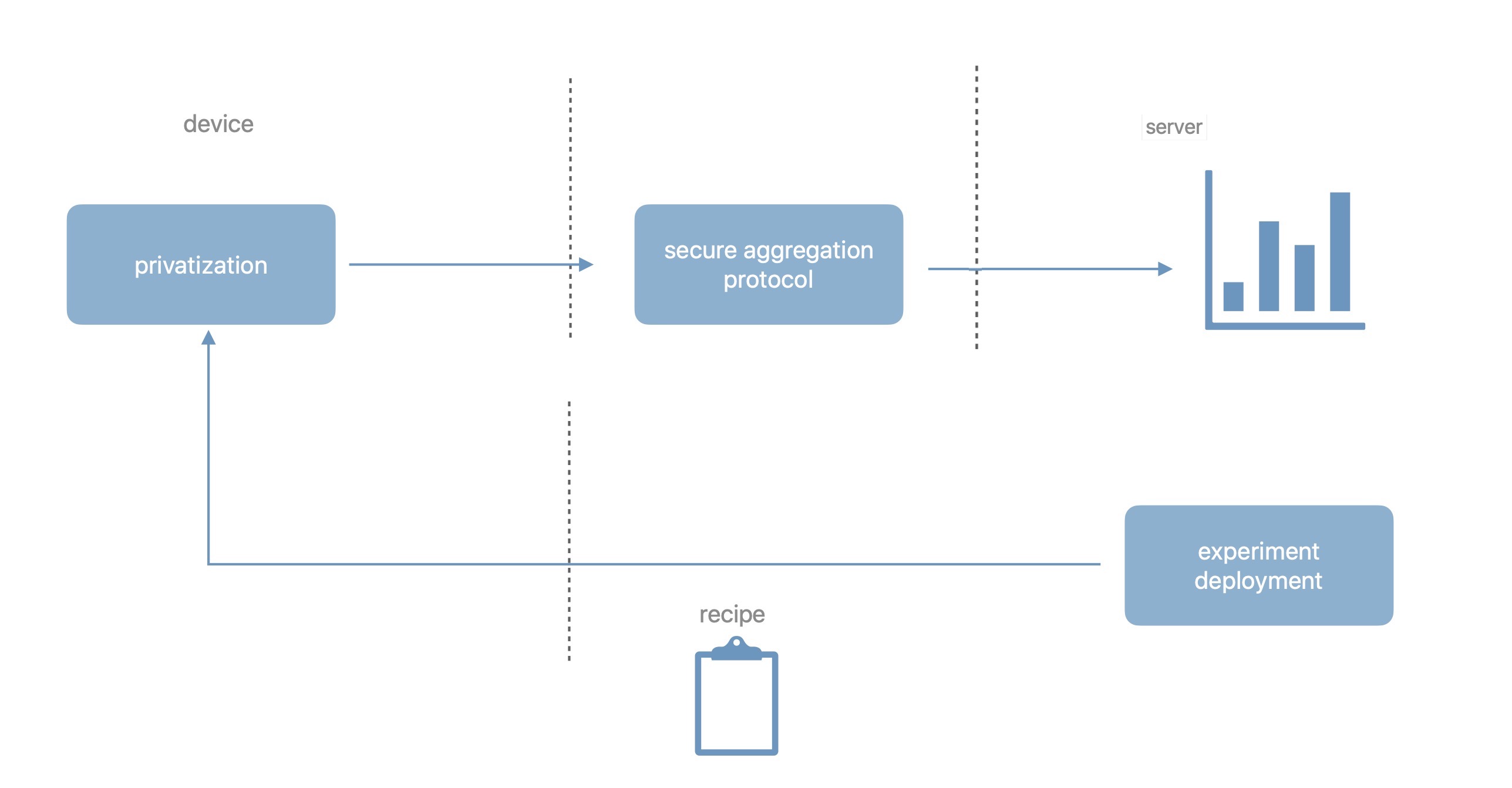}
	\caption{Proposed system architecture. Histogram query results are protected by differential privacy constraints. Reports from multiple devices are aggregated using an aggregation protocol so that the server only sees the aggregate report from a cohort of a pre-specified minimum size. 
	}
		\label{fig:system}
	\end{center}
\end{figure}

                \lstset{language={},literate={``}{\textquotedblleft}1, xleftmargin=.3in, numberstyle=\tiny,xrightmargin=-.1in}

Prior to being queried, an on-device framework should be used to provide secure storage for the data. This framework should be optimized for fast sequential access to user data, since this data is generated by user actions we sometimes call these data generation events simply ``events". Each event is donated to an approved data stream with unique identifier, and fixed schema. 
For example, when a user types a phrase using their keyboard,  an event would be added to the corresponding ``keyboard" stream.

An access control mechanism can be used to ensure that this data store is readable only by queries which are approved to access this data. Further, a retention policy assures that each event is stored on the device for a limited, approved period of time.
Figure \ref{fig:device-system} contains an example of a potential on-device architecture.

\begin{figure}[H]
	\begin{center}
			\includegraphics[scale=0.15]{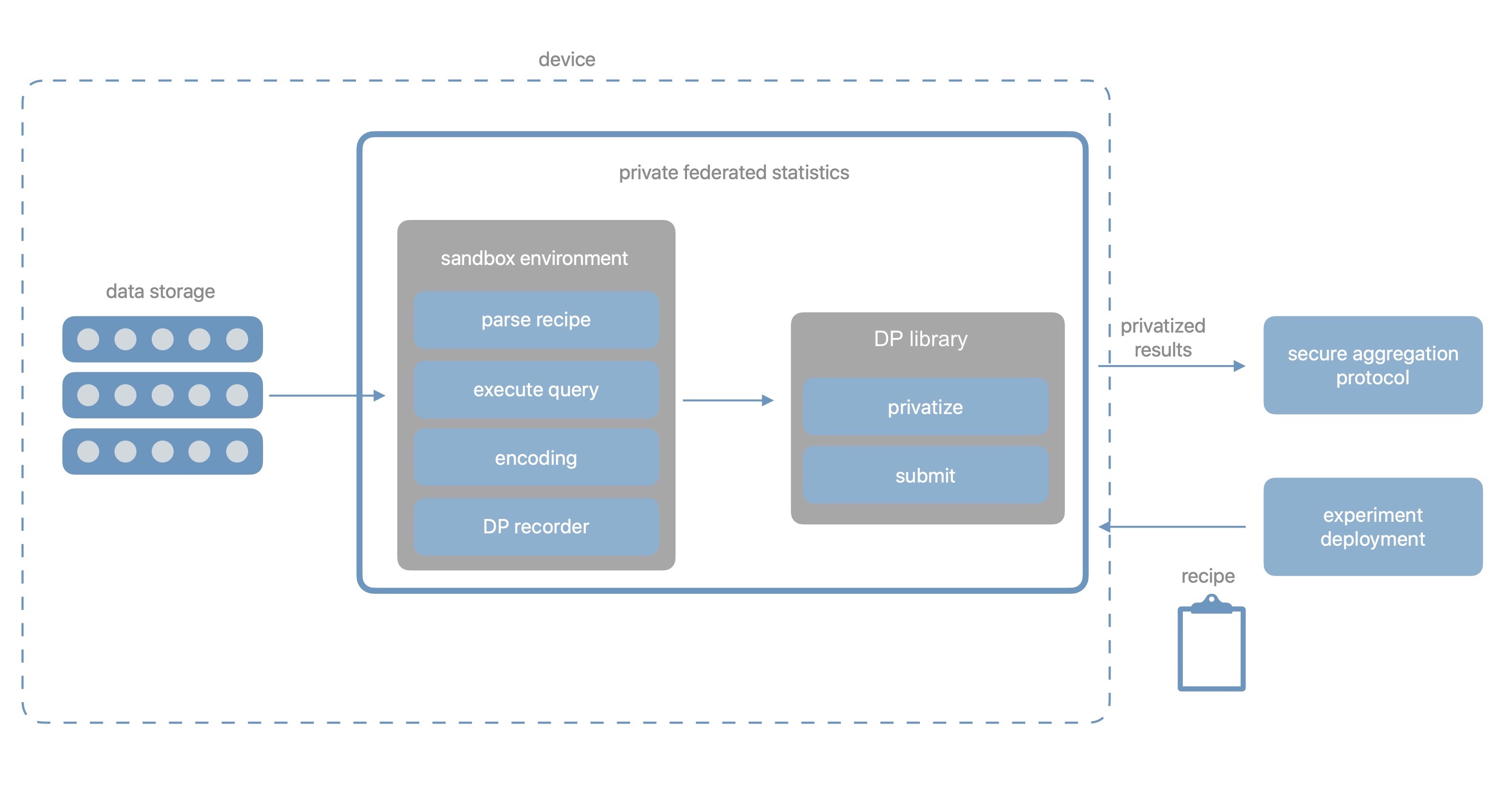}
	\caption{On device architecture example. The query and data are processed in the sandbox environment, then passed to the DP library, which contains algorithms carefully vetted to satisfy differential privacy.}
		\label{fig:device-system}
	\end{center}
\end{figure}

\subsection{Device Enrollment}\label{device_enrollment}
In order to run a query from the server, a \emph{recipe} is deployed to all opted-in devices. A recipe includes a query which is executed on the data stream, and a schema to encode the query output to a one hot vector.
For example, for the algorithm outlined in the introduction, for interactively learning n-grams of known words, one query in this analysis would ask users to send a randomly chosen 3-gram, among those they have typed, that starts with a member of a small list of popular 2-grams.\footnote{If the user does not have any 2-grams that match the requirements, then they should send a random 3-gram as sending nothing can itself be disclosive. Further, it is sufficient to send a hash of the list of popular 2-grams to the devices.}
In order to enable learning popular 3-grams, we send the list of popular 2-grams in a prefix tree data structure, as well as set of known tokens for completing this prefix tree. If a token is not in this list, it will be tokenized to a special token \emph{OOV} which denotes an \emph{Out of Vocabulary} token.

In this example, the recipe sent from the server needs to contain several pieces of information. Firstly, the recipe should contain the query itself, and the requested local and cohort aggregate privacy budgets for this query. In addition, a scheme to encode the query output as a one-hot vector should be included. This ensures the query output is an appropriately formed input for the types of local randomizers that arise in differentially private histogram estimation. This scheme is often defined by the tuple of features we intend to learn, and a bucketing rule for each feature. For an example, see Appendix~\ref{appendix:deviceenrollment}.

If any data fields accessed by the query are determined to be non-sensitive, then these can also be specified in the recipe. For example, when learning n-grams of known words, the locale (device language, eg. american english) may be deemed non-sensitive. The device should keep a list of approved non-sensitive data fields hard-coded as part of the device OS. The recipe may also contain meta information such as a recipe identifier or version number, which can be attached to the privatized query response when being sent to the server.

\subsection{Query Class Verification}\label{queryverification}

As mentioned earlier, the analysis is restricted to only ask queries that lie within the approved query class. In the non-adaptive setting, this query class may be a singleton set (a set that contains only one element), and the architecture mainly serves to co-ordinate from the server when to run the query. More generally, a broader set of queries may be allowed. Any query class should be supported by a query class verifier; a piece of code that runs on device and can verify that the query given in a recipe is a member of the query class.
When the number of allowed queries is small, this can be achieved by simply listing the set of allowable queries.

As an example of a broader query class, suppose that the query class is specified by an allowed set of data fields, and any query that only accesses those data fields is a permitted query. One could verify that a query belonged to this query class by configuring the underlying database management system (DBMS) to restrict access to the allowed list of fields. Upon execution of the query, the DBMS will return an error indicating that the query was denied if a non-allowed field is referenced.

\subsection{Privacy Budget Verification}\label{sec:privacy-budget}
All analyses run on the system we propose satisfy strong differential privacy guarantees.
Each analysis must satisfy a cohort aggregate privacy budget that is shared over all the queries contained in that analysis. Recall that the cohort aggregate privacy budget is the privacy guarantee of the output of an aggregation protocol with the additional functionality that it can enforce a minimum cohort size.  Additionally, each query must satisfy a specified local privacy budget.
Since queries within an allowable query class may access data types with varying sensitivity, we also employ a more fine-grained differential privacy budget for individual queries. Each data field is assigned an allowed local privacy budget, an allowed cohort aggregate privacy budget and an allowed number of types that data type may be accessed per analysis. Any query must satisfy the most stringent privacy conditions of any data field that is accessed by that query.

Table~\ref{ondevicebudgetaccountant} summarizes hypothetical parameter settings for the proposed on-device budget accountant that tracks and limits the differential privacy parameter for each data field, for queries related to the use of the keyboard. In this example, a team interested in improving next word prediction may request access to $n$-grams typed with the keyboard, as well as metrics related to the performance of the currently deployed keyboard model like model perplexity. More sensitive fields like bucketed age may be useful for assessing how performance varies among subpopulations. By using this fine grained privacy budget account, we can ensure that not only is the total differential privacy parameter bounded, but particularly sensitive information is further protected.
When a query is sent to the device, a series of privacy budget checks ensure that the device has enough remaining differential privacy budget to run the query with the desired differential privacy parameters:

\begin{itemize}
\item \textbf{Check 1:}  Check that the analysis has enough differential privacy budget remaining and has not exceeded its allowed number of queries per Table~\ref{usecaseprivacybudgets}. Formally, $\eps+\eps_{\rm Used}\le \eps_{\rm Allowed}$ and  $k_{\rm Used}+1\le k_{\rm Allowed}$.
\item \textbf{Check 2:} Check that data fields used in the query have enough differential privacy budget remaining and have not exceeded their allowed number of queries per Table~\ref{datatypeprivacybudgets}. Formally, for every data field accessed by the query, $\eps_0\le \eps_{\rm Allowed, Data\; field}^{\rm local}$, $\eps+\eps_{\rm Used, Data\; field}\le \eps_{\rm Allowed, Data\; field}$ and $k_{\rm Used, Data\; field}+1\le k_{\rm Allowed, Data\; field}$. This check requires a complete and accurate list of data fields that are accessed by the query. 
\item \textbf{Check 3:} The cohort-aggregate differential privacy guarantee $\eps$ will be achieved by the aggregation of at least $m$ $\eps_0$-DP local reports.
\end{itemize}

If all checks are passed, the device runs the query. The differential privacy budget accountant is updated to reflect the privacy loss incurred by running the query. Formally, $\eps_{\rm Used} \mapsto \eps_{\rm Used}+\eps$, $k_{\rm Used} \mapsto k_{\rm Used}+1$, and for every data field accessed by the query $\eps_{\rm Used, Data\; field} \mapsto \eps_{\rm Used, Data\; field}+\eps$, $k_{\rm Used, Data\; field} \mapsto k_{\rm Used, Data\; field}+1$. If any of the checks fail, the device does not run the query. While we have described the privacy checks as using the basic composition theorem, these can naturally be adapted to the use of advanced composition theorem, or other privacy accounting methods.

While we do not formally track the local privacy budgets, the limit on the allowed number of reports imposes an upper bound of $\eps_{\rm Allowed}^{\rm local} \times k_{\rm Allowed}$ on the local privacy loss for the analysis, and a potentially stronger upper bound of $\eps_{\rm Allowed, Data\; field}^{\rm local} \times k_{\rm Allowed, Data\; field}$ on the local privacy loss per data field.

\begin{center}
\begin{table}
\begin{subtable}{1\textwidth}
\begin{tabular}{  | L{3.5cm}| L{3.5cm} | L{3.5cm} | L{3.5cm} | }
  \hline
 \textbf{\small Allowed Cohort Aggregate Privacy Budget $\eps_{\rm Allowed, Analysis}$} & \textbf{\small Used Cohort Aggregate Privacy Budget $\eps_{\rm Used, Analysis}$} & \textbf{\small Allowed Number of Reports $k_{\rm Allowed, Analysis}$} & \textbf{\small Number of Reports Used $k_{\rm Used, Analysis}$} \\
  \hline\hline
0.5 & 0 & 1 & 0\\
  \hline
\end{tabular}
\caption{Example of Maximum Privacy Budgets per Analysis allocations. Approved analyses must abide by their own approved privacy budgets. These budgets are adjusted to suit the sensitivity of the data used. The number of local reports that an analysis can receive from a given device is also limited.}\label{usecaseprivacybudgets}
\end{subtable}
\begin{subtable}{1\textwidth}
\begin{tabular}{ | L{1.7cm} | L{2.5cm}| L{2.5cm} | L{2.5cm} | L{2.5cm} | L{2.5cm} |}
  \hline
\textbf{\small Data Field} & \textbf{\small Allowed Local Privacy Budget $\eps_{\rm Allowed, Data\; field}^{\rm local}$} & \textbf{\small Allowed Cohort Aggregate Privacy Budget $\eps_{\rm Allowed, Data\; field}$} & \textbf{\small Used Cohort Aggregate Privacy Budget $\eps_{\rm Used, Data\; field}$} & \textbf{\small Allowed Number of Reports $k_{\rm Allowed, Data\; field}$} & \textbf{\small Number of Reports Used $k_{\rm Used, Data\; field}$} \\
  \hline\hline
  $n$-gram & 5 & 1 & 0 & 1 & 0\\
  \hline
  Bucketed age & 2 & 0.3 & 0 & 1 & 0 \\
  \hline
  Model perplexity & 8 & 1 & 0 & 1 & 0\\
  \hline
\end{tabular}
\caption{Example of Maximum Privacy Budgets per Data Field allocations. For each data field we bound the total privacy loss of any query that access that data field. These budgets are adjusted to suit the sensitivity of the data field. }\label{datatypeprivacybudgets}
\end{subtable}

\caption{ Example of a potential on-device differential privacy budget accountant for keyboard analysis.}\label{ondevicebudgetaccountant}
\end{table}
\end{center}

\subsection{Privatization}
Once it has been established that that there is enough privacy budget remaining to run the query, the device uses a local randomizer to respond to the query. We will describe one possible randomizer choice, private one-hot encoding (OHE), in more detail in the appendix. Briefly, private OHE involves each user representing their data as a one-hot encoded binary array, then flipping each coordinate of this array independently with some pre-specified probability. After local privatization, the randomizer output is sent to the aggregation protocol. For instance, if the aggregation is based on PRIO~\cite{Prio}, the privatized record would be secret-shared into shares that would be encrypted using the public keys of non-colluding servers. Finally, the query, privatized record, and any function of the privatized record that leaves the device (e.g. shares of the privatized record) can be logged to the device for auditability.

\subsubsection{Aggregation}
In order to build a system with strong privacy guarantees, the aggregation protocol should guarantee that it does not reveal anything about the local reports except their sum. However, our overall architecture is oblivious to the particular aggregation protocol. One option is to use a multiple server architecture as in PRIO~\cite{Prio}. Such a system with two servers has been used in earlier systems~\cite{ENPA:2021, enpatalk, firefox}. Each device secret shares their local report into two additive shares, which are sent to two non-colluding servers.
Each server aggregates the shares it receives within a time window, and conditioned on the batch size being sufficiently large, publishes its aggregate to a third server. Individual shares are only stored for this time window and are deleted right after the batch aggregation job finishes. The third server aggregates all these results, and returns the result to the analyst. By performing the aggregation in this manner, we ensure that no party has access to any individual local report.

\section{Risk Analysis}

In this section, we will address how the system proposed in this paper mitigates some of the key risks that might arise when implementing a system that allows analysts to ask queries from the server. In this work, we have assumed an implementation of the aggregation protocol that guarantees that the adversary can only get access to aggregates of at least B users, where the batch size B is routed through, and logged on, the device. One can imagine multiple ways to implement such a functionality, and different implementations would entail their own risk analyses.  In this section, we discuss risks and mitigations for the whole system assuming the aggregation protocol with this functionality is implemented securely. Below is a list of potential risks, and how the system mitigates against them.

\begin{itemize}
\item \textbf{Flexible queries could allow analysts to target small populations.} This risk is mitigated by the fact that differential privacy protects users against individual privacy loss \cite{DR14}, and protects against information being leaked about small populations. The main protection is from the aggregate differential privacy guarantee, with the local differential privacy guarantee as an additional guarantee. Further, an aggregation protocol with the additional functionality that it can enforce a lower bound on the cohort size can ensure that only aggregate information is released.

We note that this protection depends on the cohort containing a sufficiently large number of \emph{honest} users. Thus, the aggregation protocol must mitigate the risk of sybil attacks. Additionally, observe that the targeting criteria (the features of users that a query attempts to learn) being public allows for auditability, which further mitigates the targeting risk.

\item \textbf{Broad query classes could allow for accidental approval of sensitive queries.} This risk is mitigated by allowing the restrictiveness (or broadness) of the query class to be tuned to suit the setting. The query class
can be vetted for concerns other than individual privacy.
\item \textbf{Flexible queries from server may allow an analyst to set their own privacy budget.}
This risk is mitigated by strict privacy budgeting on-device that ensures that privacy loss is carefully controlled. Local privacy budgets are stored and enforced on-device.
	Cohort Aggregate privacy budgets are routed through the device and enforced by the aggregation protocol. The device only needs to trust the aggregation protocol to enforce the minimum cohort size.
\item \textbf{Flexible queries may allow an analyst to ask questions that are not part of the approved query class.}	This risk is mitigated by query verification on-device that ensures only queries in the approved query class are answered by the device. Query logs on-device can also provide further transparency about the queries being asked and answered by the device.
\end{itemize}

\section{Conclusion and Open Problems}

In this work, we proposed a system for enabling interactive federated statistics with strong privacy guarantees. The privacy guarantees are supported by an aggregation protocol that ensures that only aggregate statistics, with cohort aggregate privacy guarantees can be seen by any analyst. This aggregate privacy guarantee is reinforced by local privacy guarantees that are enforced on-device. This architecture allows the user device to maintain control over the privacy assurance and ensures that inadvertent error by an analyst leads to utility loss, \emph{not} privacy loss.

It has been shown in a number of prior works~\cite{Abadi:2016, McMahanRT018, Bagdasaryan:2022} that in the federated setting, private interactive algorithms (enabled by allowing analysts to ask adaptive queries from the server) can greatly improve the accuracy of the private statistics learned. Further research is needed to expand the class of statistics that we can learn in this framework. In practical deployments, user's can have multiple data points.
This additional data can be utilized by adaptive algorithms, for example, the algorithm described in the introduction for learning new phrases can potentially gain from carefully selecting the next data point that the user communicates.
Additionally, there can be heterogeneity in both the quantity and distribution of user data. Heterogeneity poses additional challenges for analysts both in terms of algorithm design and interpretation of results. How to best utilize user data in the multiple-data-points-per-user setting is still an active, and important, area of research.

\section*{Acknowledgements}
We thank Úlfar Erlingsson, Julien Freudiger for their guidance during the early stages of this project; Rohith Prakash, Noah Jacobs, Andrew Cherkashyn, Mona Chitnis, Roman Holenstein, Gaurav Chandalia, Chiraag Sumanth, Mayank Yadav, Abhishek Bhowmick, JP Lacerda, Yuan Li, Zhong Wang, Michael Scaria, Kristine Guo, Oliver Chick, Hwasung Lee, Yi Sheng Chan for their help with this effort.

\bibliographystyle{abbrv}
\bibliography{reference}

\appendix

\section{Asymmetric Private One-Hot Encoding}\label{sec:OHE}

\emph{Asymmetric private one-hot encoding} outputs a histogram of counts over a domain $\cD$, for a dataset consisting of $n$ individuals. This algorithm is simple and has high accuracy, but has high communication complexity when the data domain is large. At a high level, one hot encoding is composed of a client-side algorithm and a server-side algorithm.  The client-side algorithm is an $\epsilon$-DP local randomizer and ensures that the data that leaves the user's device is $\epsilon$-locally differentially private. The server side algorithm simply aggregates the local responses, and de-bias the result. This algorithm was proposed and used in RAPPOR \cite{erlingsson2014rappor}.  Let $f:\cD \rightarrow \mathbb{R}$ be an oracle that gives the frequency of an element in $\cD$ from a given dataset $\{d^{(1)},\cdots, d^{(n)}\}$. We design a differentially private algorithm $\A$ that outputs an oracle $\tilde f: \cD\rightarrow \mathbb{R}$ which is an unbiased estimate of $f(d)$, that is $\forall d\in\cD, \E \left[ \tilde f(d) \right] = f(d)$, and has small variance.

\subsection{Client-Side Algorithm}

The client encodes their bit as a vector of length $|\cD|$, which is 1 in the position of the client's data, and 0 elsewhere. They then independently send the value of each coordinate using a standard differentially private algorithm that is an asymmetric version of randomized response. Let $\texttt{Bernoulli}(p)$ be a Bernoulli random variable which is $1$ with probability $p$ and $0$ with probability $1-p$.
Algorithm~\ref{algclient:OHE} gives details of the client-side algorithm. 

\begin{algorithm}[htb]
	\caption{Client-Side $\AclientOHE$} \label{algclient:OHE}
	\begin{algorithmic}[1]
		\REQUIRE Data element $d \in \cD$; $\epsilon$.
		\STATE Initialise $v, \tilde{v} \in \{0\}^{\cD}$
		\STATE $v_d = 1$ \label{datadependence}
		\FOR{$d'\in\cD$}
			\STATE $\tilde{v_{d'}} \sim \begin{cases} \texttt{Bernoulli}(1/2) & \text{if } v_{d'}=1 \\ \texttt{Bernoulli}(1/(e^{\eps}+1)) & \text{if } v_{d'}=0 \end{cases}$
		\ENDFOR
		\RETURN $\tilde{v}$
	\end{algorithmic}
\end{algorithm}

\begin{thm}[Privacy guarantee] \label{privacyOHE}
For any $\epsilon>0$, $\AclientOHE$ is $\epsilon$-locally differentially private in the replacement model.
\end{thm}

\begin{proof}[Proof of Theorem~\ref{privacyOHE}]
Let $d, d'\in\cD$ be two data points and $\eps>0$. Therefore, for any $(v_{d''})_{d''\in\cD}\in\{0,1\}^{\cD}$, 
\begin{align*}
\frac{\Pr(\AclientOHE(d)=(v_{d''})_{d''\in\cD})}{\Pr(\AclientOHE(d)=(v_{d''})_{d''\in\cD})} &= \frac{\Pr({\rm Bernoulli}\left(1/2\right)=v_d)\prod_{d''\in\cD, d''\neq d}\Pr({\rm Bernoulli}\left(\frac{1}{e^{\eps}+1}\right)=v_{d''})}{Pr({\rm Bernoulli}\left(1/2\right)=v_{d'})\prod_{d''\in\cD, d''\neq d'}\Pr({\rm Bernoulli}\left(\frac{1}{e^{\eps}+1}\right)=v_{d''})}\\
&= \frac{\Pr({\rm Bernoulli}\left(1/2\right)=v_d)\Pr({\rm Bernoulli}\left(\frac{1}{e^{\eps}+1}\right)=v_{d'})}{\Pr({\rm Bernoulli}\left(\frac{1}{e^{\eps}+1}\right)=v_{d})\Pr({\rm Bernoulli}\left(1/2\right)=v_{d'})}\\
&= \frac{\Pr({\rm Bernoulli}\left(\frac{1}{e^{\eps}+1}\right)=v_{d'})}{\Pr({\rm Bernoulli}\left(\frac{1}{e^{\eps}+1}\right)=v_{d})}\\
&\in [e^{-\eps}, e^{\eps}]
\end{align*}
\end{proof}

\subsection{Server-Side Algorithm}

The server-side algorithm $\AserverOHE$ simply aggregates the local reports, and de-bias the resulting noisy histogram.

\begin{algorithm}[htb]
	\caption{Server-Side $\AserverOHE$} \label{algserver:OHE}
	\begin{algorithmic}[1]
		\REQUIRE $\{\tilde{v}^{(1)}, \cdots, \tilde{v}^{(n)}\}$; $\epsilon$.
		\STATE Construct $\tilde{f}:\cD\to\mathbb{R}$ where for all $d\in\cD$
		\[\tilde{f}(d) = \frac{2(1+e^{\epsilon})}{e^{\epsilon}-1}\sum_{i=1}^n \tilde{v}^{(i)}_d-\frac{n}{e^{\epsilon}-1}\]
		\RETURN $\tilde{f}$
	\end{algorithmic}
\end{algorithm}

\begin{thm}[Utility guarantee]\label{OHEutilityguarantee}
Let $\epsilon>0$, $d\in\cD$, and $\tilde{f}(d)$ be the $d$-th coordinate of the output of $\AserverOHE(\tilde{v}^{(1)}, \cdots, \tilde{v}^{(n)}; \epsilon)$ where each $\tilde{v}^{(1)}$ is reported using $\AclientOHE$. If $f(d)$ is the true frequency of $d$ then,
\[\mathbb{E}[\tilde{f}(d)]=f(d)\]
\[\var{\tilde{f}(d)} = n\cdot \frac{4e^{\epsilon}}{(e^{\epsilon}-1)^2}+f(d).\]
\end{thm}

\begin{proof}[Proof of Theorem~\ref{OHEutilityguarantee}]
Since each coordinate is independent, we can deal with each coordinate separately. Now, given $d\in\cD$ and $i\in[n]$, let $v^{(i)}_d=1$ if the $i$th persons data point is $d$, and 0 otherwise. Then $\mathbb{E}[\tilde{v}^{(i)}_d]=1/2$ if $\tilde{v}^{(i)}_d=1$ and $1/(e^{\eps}+1)$ if $\tilde{v}^{(i)}_d=0$.
\begin{align*}
\mathbb{E}[\tilde{v}^{(i)}_d] &= \frac{1}{2}v^{(i)}_d+\frac{1}{e^{\eps}+1}(1-v^{(i)}_d)\\
&= \frac{1}{e^{\eps}+1}+\frac{e^{\eps}-1}{2(e^{\eps}+1)}v_d^{(i)}
\end{align*}
Since $\tilde{f}(d)=\sum_{i=1}^n v_d^{(i)}$, the unbiasedness follows from the linearity of expectation. Also, \begin{align*}
\var{\tilde{v}_d^{(i)}} &= \frac{1}{4}v_d^{(i)}+\frac{e^{\eps}}{(e^{\eps}+1)^2}(1-v_d^{(i)})\\
&= \frac{e^{\eps}}{(e^{\eps}+1)^2}+\frac{(e^{\eps}-1)^2}{4(e^{\eps}+1)^2}v_d^{(i)}
\end{align*}
Since each $\tilde{v}_d^{(i)}$ is independent, 
\begin{align*}
\var{\tilde{f}(d)} &= \frac{4(1+e^{\eps})^2}{(e^{\eps}-1)^2}\sum_{i=1}^n\left( \frac{e^{\eps}}{(e^{\eps}+1)^2}+\frac{(e^{\eps}-1)^2}{4(e^{\eps}+1)^2}v_d^{(i)}\right)\\
&= n\frac{4e^{\eps}}{(e^{\eps}-1)^2}+f(d)
\end{align*}
\end{proof}

\section{Device Enrollment - Detailed Example}\label{appendix:deviceenrollment}

As an example, suppose the analysts' goal is to learn the popular $n$-grams joined with user age. 
An example of a recipe used for learning a histogram of popular 3-grams, joined with age of the device owner is presented in Figure \ref{fig:recipe}. The fields in this recipe are as follows:

\begin{itemize}
\item recipe identifier: a unique identifier for this recipe. 
\item version: recipe version for compatibility. 
\item analysis identifier: provides the unique identifier which is used as the prefix for collection ids so that the server can identify payloads corresponding to this histogram query. This prefix can be part of a list hardcoded as part of the device OS as well, for further verification. Furthermore, data deemed as \emph{non-sensitive}, as well as experiment ids, could be appended to this string before submission to server.
\item query: this query is first verified and if it is an allowed query, it will be executed on the on-device data. If the query is not allowed, the system exits and no privatized record is shared with the secure aggregation protocol.
\item cohorts: if any data, relevant to the query, is determined to be non-sensitive, then it is specified here. This data could be appended to the collection id prefix provided in client identifier. For example, when learning n-grams of known words, the locale (device language like $``en\_US"$) may be deemed \emph{non-sensitive}. The device should keep a list of approved non-sensitive data fields for each analysis hardcoded as part of the device OS.

\item data content type: a scheme to encode the query output to a one hot vector. In this field, we determine the tuple of features we intend to learn, and a rule for bucketing each element of this tuple. In the example presented in Figure  \ref{fig:recipe}, we choose to build a histogram of  ({\bf age, ngrams}). 
\begin{itemize}
\item Age is bucketed based on the bucket boundaries provided in the recipe. In the example recipe in Figure \ref{fig:recipe}, the buckets for age are $\{\emph{OOV}, \emph{[20, 30)}, \emph{[30, 40]}, \cdots, \emph{[70, 80)}] $, where \emph{OOV} is added on the device for the elements which are not covered in the bucketing rule. The total number of age buckets is $(1 + 7)$.
\item  n-grams are bucketed based on the prefix tree, and the set of known words provided on the recipe. Every combination of prefix tree leaf and a known word listed in recipe is considered to be a bucket. In addition to OOV token, the device adds an \emph{end-token} to determine if a leaf in the given prefix tree is a leaf as well. For the recipe in Figure \ref{fig:recipe}, the buckets for 3-grams are $\{$\emph{OOV}, \emph{hello world end-token},\emph{hello world OOV}, \emph{hello world a}, $\cdots,$ \emph{i got world}$\}$. The total number of buckets are 
$(1 + 3 \times (1 + 1 + 9))$
\item eventually, every combination of  buckets above builds a one hot vector for ({\bf age, ngrams}). The total number of bins will be $(1 + 7) \times (1 + 3 \times (1 + 1 + 9))$
\end{itemize}

\end{itemize}

\begin{figure}[H]
	\begin{center}
			\includegraphics[scale=0.4]{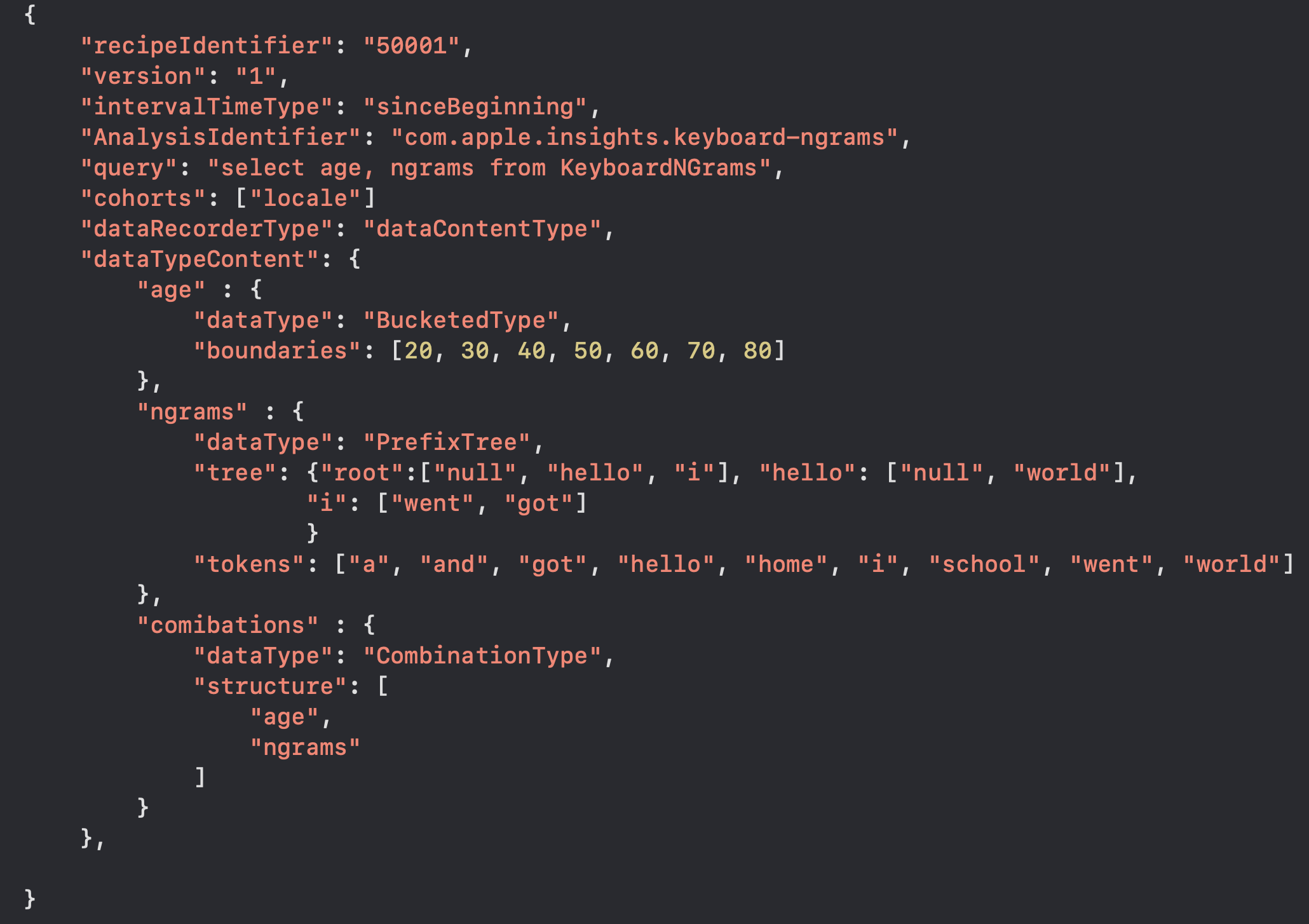}
	\caption{Sample recipe for learning join distribution of 3-grams and device owner age.}
		\label{fig:recipe}
	\end{center}
\end{figure}

\section{Deletion Differential Privacy}\label{sec:deletionDP}

In this section, we will outline how one could use deletion model of local differential privacy as the main local differential privacy, rather than the replacement model outlined in Defintions~\ref{LDP} and~\ref{localDP}. First, let us formally define the deletion model.

\begin{defn}[Local Randomizer \cite{DR14,Kasiviswanathan:2008,Erlingsson2020}]
Let $\A: \cD \to \cY$ be a randomized algorithm mapping a data entry in $\cD$ to $\cY$.  The algorithm $\A$ is an $\epsilon$-DP local randomizer in the \emph{deletion} model if there exists a reference distribution $\reference$ such that for all data entries $d \in \cD$ and all events $E\subset \cY$, we have 
\[
- \epsilon \leq \ln\left(\frac{\Pr[\A(d) \in E ]}{\Pr[\reference \in E ]}  \right)\leq \epsilon
\]
\end{defn}

The deletion model of differential privacy is closely related to the replacement model introduced in Definition~\ref{LDP}; if $\A$ is $\eps$-DP in the deletion model then $\A$ is $\eps'$-DP in the replacement model where $\eps'$ can take any value in $[\eps, 2\eps]$. Exactly where the replacement privacy guarantees lies in $[\eps, 2\eps]$ is algorithm dependent. Conversely, if $\A$ is $\eps$-DP in the replacement model then it is $\eps$-DP in the deletion model.

Differential privacy in the deletion model ensures that it is impossible to confidently determine whether an individual contributed data or simply reported a sample from the reference distribution (i.e. did not contribute data). This is slightly weaker than replacement DP, which guarantees that from the output of $\A$, it is impossible to confidently guess whether the user’s data point was $d$ or $d'$, for any pair $d$ and $d'$. If we think of the reference distribution as the "average" user behavior then, intuitively, deletion DP means it hard to distinguish the behavior of an outlier from the average behavior. Replacement DP provides the slightly stronger guarantee that it is hard to distinguish between any two outliers, even if the outliers are ``opposites". 
One can define the deletion model for local differential privacy in exactly the same manner as Definition~\ref{localDP} where the local randomizers must be $\epsilon$-DP in the deletion model.

\subsection{Symmetric Private One-Hot Encoding}\label{sec:symOHE}

In section, we will outline a locally differentially private randomizer that is appropriate for use when implementing the system we have outlined with deletion local differential privacy. \emph{Symmetric private one-hot encoding} is a slight variant on asymmetric private one-hot encoding, which was outlined in Section~\ref{sec:OHE}.

\subsubsection{Client-Side Algorithm}

The client encodes their bit as a vector of length $|\cD|$, which is 1 in the position of the client's data, and 0 elsewhere. They then independently send the value of each coordinate using a standard differentially private algorithm called randomized response.
Algorithm~\ref{algclient:OHE} gives details of the client-side algorithm. 

\begin{algorithm}[htb]
	\caption{Client-Side $\AclientsymOHE$} \label{algclient:OHE}
	\begin{algorithmic}[1]
		\REQUIRE Data element $d \in \cD$; $\epsilon$.
		\STATE Initialise $v, \tilde{v} \in \{0\}^{\cD}$
		\STATE $v_d = 1$
		\FOR{$d'\in\cD$}
			\STATE $\tilde{v_{d'}} = \begin{cases} v_{d'} & \text{with probability } \frac{e^{\epsilon}}{1+e^{\epsilon}}\\ 1-v_{d'} & \text{otherwise} \end{cases}$
		\ENDFOR
		\RETURN $\tilde{v}$
	\end{algorithmic}
\end{algorithm}

\begin{thm}[Privacy guarantee] \label{privacysymOHE}
For any $\epsilon>0$, $\AclientsymOHE$ is $\epsilon$-locally differentially private in the deletion model.
\end{thm}

\begin{proof}
The reference distribution $\reference$ is obtained by independently for each coordinate outputting 0 with probability $\frac{e^{\epsilon}}{1+e^{\epsilon}}$ and 1 otherwise. Note that given datapoint $d\in\cD$, the reference distribution matches the distribution $\AclientsymOHE(d; \epsilon)$ on every coordinate except $d$. Therefore, given $\bbv\in\{0,1\}^{\cD}$,
\begin{align*}
\frac{\Pr(\reference=\bbv)}{\Pr(\AclientsymOHE(d)=\bbv)} &= \frac{\prod_{d'\in\cD}\Pr((\reference)_{d'}=\bbv_{d'})}{\prod_{d'\in\cD} \Pr((\AclientsymOHE(d))_{d'}=\bbv_{d'})}\\
&=\frac{\Pr((\reference)_{d}=\bbv_{d})}{\Pr((\AclientOHE(d))_{d}=\bbv_{d})}\\
&=\begin{cases} 
\frac{e^{\epsilon}}{1} & \text{if } \bbv_d=0\\
\frac{1}{e^{\epsilon}} & \text{if } \bbv_d=1
\end{cases}.
\end{align*}
\end{proof}

\noindent Note that $\AclientsymOHE$ is $2\epsilon$-locally differentially private in the replacement model.

\subsubsection{Server-Side Algorithm}

The server-side algorithm $\AserversymOHE$ simply aggregates the local reports, and de-bias the resulting noisy histogram.

\begin{algorithm}[htb]
	\caption{Server-Side $\AserversymOHE$} \label{algserver:OHE}
	\begin{algorithmic}[1]
		\REQUIRE $\{\tilde{v}^{(1)}, \cdots, \tilde{v}^{(n)}\}$; $\epsilon$.
		\STATE Construct $\tilde{f}:\cD\to\mathbb{R}$ where for all $d\in\cD$
		\[\tilde{f}(d) = \frac{1+e^{\epsilon}}{e^{\epsilon}-1}\sum_{i=1}^n \tilde{v}^{(i)}_d-\frac{n}{e^{\epsilon}-1}\]
		\RETURN $\tilde{f}$
	\end{algorithmic}
\end{algorithm}

\begin{thm}[Utility guarantee]\label{OHEutilityguarantee}
Let $\epsilon>0$, $d\in\cD$, and $\tilde{f}(d)$ be the $d$-th coordinate of the output of $\AserversymOHE(\tilde{v}^{(1)}, \cdots, \tilde{v}^{(n)}; \epsilon)$ where each $\tilde{v}^{(1)}$ is reported using $\AclientsymOHE$. If $f(d)$ is the true frequency of $d$ then,
\[\mathbb{E}[\tilde{f}(d)]=f(d)\]
\[\var{\tilde{f}(d)} = n\cdot \frac{e^{\epsilon}}{(e^{\epsilon}-1)^2}.\]
\end{thm}

\begin{proof}
For each $i\in[n]$, 
\[\mathbb{E}[\tilde{v}^{(i)}_d]=\frac{e^{\epsilon}-1}{1+e^{\epsilon}}v^{(i)}_d+\frac{1}{1+e^{\epsilon}} \;\;\;\text{and}\;\;\; \var{\tilde{v}^{(i)}_d)}= \frac{e^{\epsilon}}{(e^{\epsilon}+1)^2}.\] Therefore, 
\begin{align*}
\mathbb{E}[\tilde{f}(d)] &= \frac{1+e^{\epsilon}}{e^{\epsilon}-1}\sum_{i=1}^n \mathbb{E}[\tilde{v}^{(i)}_d]-\frac{n}{e^{\epsilon}-1}\\
&= \frac{1+e^{\epsilon}}{e^{\epsilon}-1}\sum_{i=1}^n \left(\frac{e^{\epsilon}-1}{1+e^{\epsilon}}v^{(i)}_d+\frac{1}{1+e^{\epsilon}}\right)-\frac{n}{e^{\epsilon}-1}\\
&= \sum_{i=1}^n v^{(i)}_d
\end{align*}
and 
\begin{align*}
\var{\tilde{f}(d)} &= \var{\frac{1+e^{\epsilon}}{e^{\epsilon}-1}\sum_{i=1}^n \tilde{v}^{(i)}_d}\\
&= \sum_{i=1}^n \left(\frac{1+e^{\epsilon}}{e^{\epsilon}-1}\right)^2 \var{ \tilde{v}^{(i)}_d}\\
&= \sum_{i=1}^n \left(\frac{1+e^{\epsilon}}{e^{\epsilon}-1}\right)^2 \frac{e^{\epsilon}}{(e^{\epsilon}+1)^2}\\
&=n\cdot \frac{e^{\epsilon}}{(e^{\epsilon}-1)^2}
\end{align*}
\end{proof}

\subsubsection{Privacy Amplification by Aggregation Bounds for Private One-Hot Encoding}

Theorem~\ref{amplificationbyaggregation} gave an upper bound on the privacy parameters for aggregating $n$ local reports, each sent using a replacement local DP randomizer. We also discussed that an upper bound could be computed numerically based on algorithms presented in~\cite{Feldman:2022}. When using local DP randomizers that are $\eps_0$-DP in the deletion model, one can use Theorem~\ref{amplificationbyaggregation} after using the conversion from deletion DP privacy parameters to replacement DP privacy parameters. However, this approach is suboptimal for many local randomizers, including symmetric private OHE. In this section, we will outline a numerical method for computing an upper bound on the cohort aggregate privacy guarantee when aggregating $n$ copies of symmetric private OHE. We'll be using the privacy amplification in terms of an alternate version of differential called R\'enyi differential privacy.

\begin{defn}[R\'enyi divergence]
For two random variables $P$ and $Q$, the R\'enyi divergence of $P$ and $Q$ of order $\alpha>1$ is \[D^{\alpha}(P\|Q) = \frac{1}{\alpha-1}\ln \E_{x\sim Q}\left[\left(\frac{P(x)}{Q(x)}\right)^{\alpha} \right] .\]
For $\alpha =1$, $D^{1}(P\|Q) = \mbox{KL}(P\|Q) = \E_{x\sim P}\left[\ln\left(\frac{P(x)}{Q(x)}\right)\right]$.
\end{defn}

Firstly, we'll need a method for computing the general amplification by aggregation bounds for binary randomized response. Binary randomized response is defined by \[\texttt{2RR}(d; \eps_0) = \begin{cases} d & \text{with probability} \frac{e^{\eps_0}}{e^{\eps_0}+1} \\ 1-d & \text{otherwise.} \end{cases}\] The algorithm $\AclientsymOHE$ simply performs binary randomized response individually on each coordinate. Define $\rho_{\texttt{2RR}}(n, \alpha, \eps_0)$ to be a function such that for any two databases $D, D'\in\{0,1\}^n$ that differ on the data of a single individual, \[D^{\alpha}\left(\sum_{i=1}^n \texttt{2RR}(d^{(i)}; \eps_0), \sum_{i=1}^n \texttt{2RR}(d^{'(i)}; \eps_0)\right)\le \rho_{\texttt{2RR}}(n,\alpha,\eps_0)\]
For any given $n, \alpha$ and $\eps_0$, the function $\rho_{\texttt{2RR}}(n, \alpha, \eps_0$ can be computed numerically. One can also use general upper bounds on the privacy amplification by aggregation in terms of R\'enyi differential privacy, presented in \cite{Feldman:2022}.

\begin{thm}\label{RAPPORrenyi}
Let $\eps_0>0, \eps>0$, $\delta\in[0,1]$, $n\in\mathbb{N}$, $\alpha>1$ and $\cD$ be a finite data universe. Then, for all pairs of datasets $D, D'\in\cD^n$ that differ on the data of a single individual, \[D^{\alpha}\left(\sum_{i=1}^n \AclientsymOHE(d^{(i)}; \eps_0), \sum_{i=1}^n \AclientsymOHE(d^{'(i)}; \eps_0)\right)\le 2\rho_{\texttt{2RR}}(n, \alpha, \eps_0).\]
\end{thm}

\begin{proof}
Assume without loss of generality that $D$ and $D'$ differ on the data of the first individual. We can view each coordinate as an independent differentially private local randomizer, and the output of $\AclientsymOHE$ as the composition of these local randomizers.
Notice that the distributions of $\sum_{i=1}^n \AclientsymOHE(d^{(i)}; \eps_0)$ and $\sum_{i=1}^n \AclientsymOHE(d^{'(i)}; \eps_0)$ only differ in the $d^{(1)}$th and $d^{'(1)}$th coordinate. Therefore, we can use the composition theorem of R\'enyi differential privacy \cite{Mironov2017RnyiDP}. 
\begin{align*}
D^{\alpha}&\left(\sum_{i=1}^n \AclientsymOHE(d^{(i)}; \eps_0), \sum_{i=1}^n \AclientsymOHE(d^{'(i)}; \eps_0)\right)\\
&\le D^{\alpha}\left(\left[\sum_{i=1}^n \AclientsymOHE(d^{(i)}; \eps_0)\right]_{d^{(1)}}, \left[\sum_{i=1}^n \AclientsymOHE(d^{'(i)}; \eps_0)\right]_{d^{(1)}}\right)\\
&\hspace{0.5in} +D^{\alpha}\left(\left[\sum_{i=1}^n \AclientsymOHE(d^{(i)}; \eps_0)\right]_{d^{'(1)}}, \left[\sum_{i=1}^n \AclientsymOHE(d^{'(i)}; \eps_0)\right]_{d^{'(1)}}\right)\\
&\le 2\eps_{\texttt{2RR}}(n, \alpha, \eps_0),
\end{align*}
where the final bound holds since restricted to each coordinate, $\AclientsymOHE(d^{'(i)}; \eps_0)$ is simply binary randomized response.
\end{proof}

There are standard formulas for converting bounds on the R\'enyi divergence to privacy guarantees in terms of $(\eps, \delta)$-DP~\cite{Mironov2017RnyiDP, Kamath:2020}.
Thus, we can obtain an upper bound on the privacy amplification by aggregation when each user uses $\AclientsymOHE$ to communicate their data, by converting Theorem~\ref{RAPPORrenyi} into a bound on the privacy loss in terms of $(\eps, \delta)$-DP.

\end{document}